\documentclass{article}

\usepackage{arxiv}

\usepackage[utf8]{inputenc} % allow utf-8 input
\usepackage[T1]{fontenc}    % use 8-bit T1 fonts
\usepackage{hyperref}       % hyperlinks
\usepackage{url}            % simple URL typesetting
\usepackage{booktabs}       % professional-quality tables
\usepackage{amsfonts}       % blackboard math symbols
\usepackage{nicefrac}       % compact symbols for 1/2, etc.
\usepackage{microtype}      % microtypography

\usepackage{graphicx}
\usepackage{natbib}
\usepackage{doi}

\usepackage{amsmath,colortbl,hhline,float}
\usepackage{array,dsfont,subcaption,multirow}
\usepackage{bm}
\usepackage{arydshln}
\usepackage[usenames,dvipnames,table]{xcolor}
\usepackage{algorithm}
\usepackage{algpseudocode}
\usepackage{amssymb,verbatim,amsthm}

\newtheorem{theorem}{Theorem}
\newtheorem{lemma}{Lemma}
\newtheorem{proposition}{Proposition}

\newcommand{\bx}{\mathbf{x}} % bold x
\newcommand{\by}{\mathbf{y}} % bold y
\newcommand{\bz}{\mathbf{z}} % bold z
\newcommand{\bM}{\mathbf{M}} % bold M
\newcommand{\bI}{\mathbf{I}} % bold I
\newcommand{\bV}{\mathbf{V}} % bold V
\newcommand{\bA}{\mathbf{A}} % bold A
\newcommand{\bU}{\mathbf{U}} % bold U
\newcommand{\bw}{\mathbf{w}} % bold w
\newcommand{\bmu}{\boldsymbol{\mu}} % bold \mu
\newcommand{\bbeta}{\boldsymbol{\beta}}
\newcommand{\balpha}{\boldsymbol{\alpha}}
\newcommand{\btheta}{\boldsymbol{\theta}} % bold \theta
\newcommand{\bvtheta}{\boldsymbol{\vartheta}} % bold \vartheta
\newcommand{\bLambda}{\boldsymbol{\Lambda}} % bold \Lambda
\newcommand{\bMtil}{\widetilde{\mathbf{M}}}
\newcommand{\bVtil}{\widetilde{\mathbf{V}}}

\newcommand{\Expect}[1]{\mathbb{E} \left[{#1}\right]}

\newcommand{\Var}[1]{\mbox{Var} \left[{#1}\right]}

\newcommand\bs{\boldsymbol}

\newcommand{\norm}[1]{\left\lVert#1\right\rVert}

\newcommand{\trace}[1]{\mbox{\textsf{trace}}\left({#1}\right)}

\title{SwISS: A Scalable Markov chain Monte Carlo Divide-and-Conquer Strategy}

\author{Callum Vyner  \\
	Department of Mathematics and Statistics\\
	Lancaster University\\
	Lancaster, LA1 4YF, UK \\
	\And
	Christopher Nemeth \\
    Department of Mathematics and Statistics\\
	Lancaster University\\
	Lancaster, LA1 4YF, UK \\
	\texttt{c.nemeth@lancaster.ac.uk} \\
	\AND
    Chris Sherlock \\
    Department of Mathematics and Statistics\\
	Lancaster University\\
	Lancaster, LA1 4YF, UK \\
	\texttt{c.sherlock@lancaster.ac.uk} \\
}

\renewcommand{\shorttitle}{\textit{SwISS}}

\hypersetup{
pdftitle={SwISS},
pdfauthor={Callum Vyner, Chistopher Nemeth, Chris Sherlock},
pdfkeywords={Markov chain Monte Carlo; divide-and-conquer; parallel MCMC; big data}
}

\begin{document}
\maketitle

\begin{abstract}
Divide-and-conquer strategies for Monte Carlo algorithms are  an increasingly popular approach to making Bayesian inference scalable to large data sets. In its simplest form, the data are partitioned across multiple computing cores and a separate Markov chain Monte Carlo algorithm on each core targets the associated partial posterior distribution, which we refer to as a \textit{sub-posterior}, that is the posterior given only the data from the segment of the partition associated with that core. Divide-and-conquer techniques reduce computational, memory and disk bottle necks, but make it difficult to recombine the sub-posterior samples. We propose SwISS: \textit{Sub-posteriors with Inflation, Scaling and Shifting}; a new approach for recombining the sub-posterior samples which is simple to apply, scales
to high-dimensional parameter spaces and accurately approximates the original posterior distribution through affine transformations of the sub-posterior samples. We prove that our transformation is asymptotically optimal across a natural set of affine transformations and illustrate the efficacy of SwISS against competing algorithms on synthetic and real-world data sets.
\end{abstract}

% keywords can be removed

\section{Introduction}
\label{sec:MCint}
Markov chain Monte Carlo (MCMC) algorithms are widely used 
within Bayesian modelling to sample from the often intractable
posterior distribution. These techniques are widely applicable and only
require point-wise evaluation of the posterior density. One of the
potential drawbacks of MCMC algorithms is their lack of scalability. The computational
cost of MCMC is typically linear in the amount of data and can be prohibitive for large data sets, both in computational cost and storage.

In settings with large data sets, or where the model is computationally expensive, evaluating the posterior at every iteration of the MCMC algorithm may be infeasible. Strategies to overcome this include data subsampling \citep{Welling2011}, \citep{Baker2017, baker2019}, \citep{nemeth2021stochastic}, where only a subset of the data is used at each MCMC iteration, or delayed acceptance \citep{sherlock2017adaptive,Quiroz2015a}, where the Metropolis--Hastings accept-reject probability is replaced with a cheaper approximation to the true posterior and the full data posterior is evaluated less frequently.

% talk about map-reduce; aws; embarrisingly parallel
In situations where it is possible to easily parallelise computation in a MapReduce framework, or through cloud-computing infrastructure such as Amazon Web Services, then statistical modelling becomes easily scalable to large data sets. However, applying this approach in practice using algorithms such as MCMC, which are designed to work in serial rather than parallel, is challenging. In this paper, we consider the divide-and-conquer strategy to circumvent the computational bottleneck of MCMC, where the data are partitioned into batches, and each batch is stored on a separate computer core. MCMC is then applied independently on each data batch and posterior samples from each computer are combined to form an accurate approximation of the full posterior, \emph{i.e.}, the posterior that would have been obtained using the full data set.

The main challenge with divide-and-conquer approaches for MCMC lies in the merge step. A range of approaches has been considered in the literature such as: the use of weighted averages of the batch samples \citep{scott2016bayes}; kernel density estimation \citep{Neiswanger2014}; Gaussian process approximations \citep{Nemeth2016}; finding the Wasserstein barycenter of different measures \citep{Srivastava2014}, the geometric median of batch samples \citep{Minsker2014}, as well as using a post-MCMC importance sample \citep{LISA2016}, to name a few.

One of the most popular algorithms in the literature is the
consensus Monte Carlo algorithm \citep{scott2016bayes}, which approximates the full
posterior using a weighted average of sub-posterior samples. The consensus
approach is computationally cheap to apply, does not require tuning and scales well to
high-dimensional parameter spaces. It is also analytically exact in the case of
Gaussian sub-posteriors, but can produce poor approximations when the
sub-posteriors are non-Gaussian (see Section \ref{subsec:MCgeom}).

In this paper we propose SwISS, an algorithm that is as fast as the consensus algorithm, is exact in the Gaussian case, does not require tuning, and which scales well to
high-dimensional posterior distributions. However, in the case of non-Gaussian sub-posteriors, it can
produce more accurate posterior approximations than the consensus algorithm.
Unlike the consensus approach, SwISS does not merge samples but instead applies a transformation to the posterior samples that are generated from a stochastic approximation of the full posterior. As in \cite{LISA2016}, we refer to this stochastic approximation as the \textit{inflated sub-posterior}, which is the posterior density, conditional on a subset of the data, raised to a positive power. Inflating the sub-posterior in this manner has the effect of approximately preserving the shape of the posterior density conditional on the full data. Affine transformations (shift and re-scale) are applied to each batch of sub-posterior samples to form an approximate sample from the full posterior. This is a generalisation of the algorithm of \cite{Wu2017} which simply shifts each sub-posterior, with no further correction and hence performs poorly when the sub-posterior variances differ substantially. There are many different affine transformations that produce a sample from the true posterior when the sub-posteriors are Gaussian; we provide theoretical support for our particular choice, showing that, in a natural sense, it is optimal amongst the set of transformations that are exact in the Gaussian case.

The paper is organised as follows: Section \ref{sec:MCbac} provides an introduction to divide-and-conquer MCMC, covering the notation for posterior and sub-posterior densities. In Section \ref{sec:MCmet} we introduce our proposed algorithm, SwISS, and provide supporting theoretical results and pseudo-code for implementation. Section \ref{sec:MCres} covers the numerical performance of SwISS and is compared against other popular divide-and-conquer algorithms from the literature. Finally, Section \ref{sec:conclusions} gives a summary of the contributions from the paper.

\section{Preliminaries}
\label{sec:MCbac}

Let $f(\by|\btheta)$ be the likelihood for a statistical model,
parameterised by $\btheta\in \mathbb{R}^d$, for a data set $\by=\{y_1,y_2,\ldots,y_n\}$ of length $n$.
Let $\pi_0(\btheta)$ denote the prior density for the parameter vector $\btheta$, then our posterior density is, up to a constant of proportionality,
\begin{equation}
  \label{eq:full_post}
  \pi(\btheta|\by) \propto \pi_0(\btheta)f(\by|\btheta).
\end{equation}

We assume that $\by$
can be partitioned into $B$
 batches, $\by_1,\dots,\by_B$,
such that the likelihood for the full data is the product of the
likelihoods for the individual batches, \emph{i.e.},
$f(\by|\btheta)=\prod_{b=1}^Bf_b(\by_b|\btheta)$. This is the case,
for example, when the individual data points are independent.
The posterior density for $\btheta$ given $\by$ is, up to a constant of proportionality, 
\begin{equation}
  \label{eq:post}
  \pi(\btheta|\by) \propto \pi_0(\btheta)\prod_{b=1}^B f_b(\by_b|\btheta).
\end{equation}
In the literature, there are generally two approaches to applying MCMC on batches of data. In the first approach, MCMC is applied to target a \textit{sub-posterior} density for each batch $b$, of the form 
\begin{equation}
  \label{eq:subpost}
  \pi_b(\btheta|\by_b) \propto  \pi_0(\btheta)^{\frac{1}{B}} f(\by_b|\btheta),
\end{equation}
where $b=1,\dots,B$, such that $\prod_{b=1}^B\pi_b(\btheta|\by_b)=\pi(\btheta|\by)$ as
defined in \eqref{eq:post}. 

If we assume that there are $J$ sub-posterior samples from each of the $B$ batches, which we define as $(\btheta_b^{(j)}; j \in \{1,\ldots,J\}, b \in \{1,\ldots,B\})$, then the consensus Monte Carlo algorithm \citep{scott2016bayes} gives a simple strategy for approximating the full posterior \eqref{eq:post} through a weighted average of the sub-posterior samples,
\[
\btheta^{(j)} =  \left(\sum_{b=1}^B \bw_b\right)^{-1}\sum_{b=1}^B \bw_b\btheta_b^{(j)},
\]
where the weights are typically chosen to be $\bw_b =\Var{\btheta|\by_b}^{-1}$. If each $\pi_b(\btheta|\by_b)$ is Gaussian, then the consensus algorithm produces exact samples from the full posterior.

% \note{The bit I've commented out is either too vague or wrong. You can
% certainly say that if the prior has tails that are lighter than
% exponential then annealing it can lead to an improper prior which in
% turn can lead to an improper posterior. That's the only obvious drawback.}
%The drawback of this strategy is that annelaing the prior in this way is not appropriate for Bayesian inference as the prior now depends on the data set size. Furthermore, for certain problems the prior plays an important role in constraining the model \note{add example}.

A second approach applies MCMC to each \textit{inflated sub-posterior}, where the target density for batch $b=1,\dots,B$ is
\begin{equation}
  \label{eq:inflate}
    \pi_b^B(\btheta|\by_b) \propto  \pi_0(\btheta)f(\by_b|\btheta)^B.
\end{equation}
This is a stochastic (across partitions of
the data) approximation to the full posterior $\pi(\btheta|\by) \approx \pi_b^B(\btheta|\by_b)$, and hence individual
samples from it, in a sense, are already on the same scale as samples
from the full posterior.  

%In contrast with the annealed prior approach \eqref{eq:subpost}, each
%sample from a sub-posterior is, in a sense, already on a similar scale
%to the full posterior since it is a sample from a stochastic
%approximation to it.
%the right scaleusing the inflated strategy \eqref{eq:inflate} means
%that the sub-posterior samples do not need to be averaged, and samples
%from each sub-posterior can be treated directly as approximate samples
%from the full posterior \eqref{eq:post}. CALLUM: Not all sub-posterior
%approaches use averaging.

If we assume that the data are partitioned equally across batches, then in the limit, as the amount of data in each batch $n_*=n/B$ approaches infinity, the likelihood will typically dominate the prior, so that by the Bernstein von Mises theorem,
the $b^{\mathrm{th}}$ inflated sub-posterior is $\btheta_b\sim \mathcal{N}(\hat{\btheta}_b,I_{O,b}^{-1}(\hat{\btheta}_b))$, where approximately, $\hat{\btheta}_b\sim \mathcal{N}(\btheta_0,B I_E^{-1}(\btheta_0))$ and where $I_E$ and $I_{O,b}$ are the full-data expected information and the observed information from the inflated likelihood for batch $b$, respectively, and $\btheta_0$ is the true parameter value. Hence, the difference between the expectations of the inflated sub-posteriors are $\mathcal{O}(n_*^{-1/2})$ and, since $\lim_{n\rightarrow \infty}|| I_E^{-1}I_{O,b}||=1+\mathcal{O}(n_*^{-1/2})$, the ratio of the variances of the  sub-posteriors is $1+\mathcal{O}(n_*^{-1/2})$. However, in practice, both the location and scale of the inflated sub-posteriors can vary considerably if the partitioned data sets are imbalanced (see examples in Section \ref{sec:MCres}). Our proposed algorithm, SwISS, provides a correction for the discrepancy in the variance and location of the sub-posterior approximations.

\section{SwISS Algorithm}
\label{sec:MCmet}

\begin{figure*}[t!]
\centering
\includegraphics[scale=0.4]{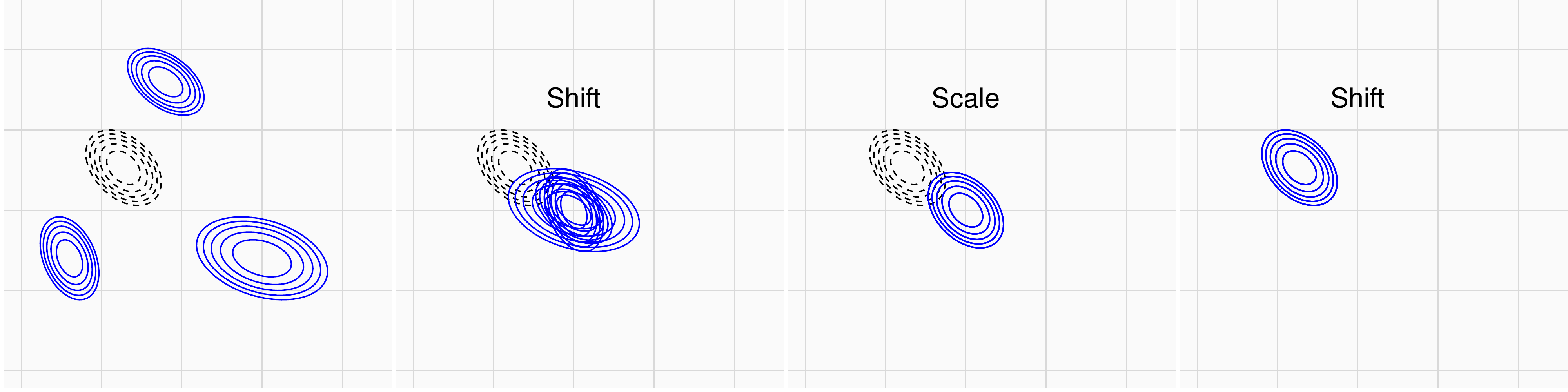}
\caption{A visual representation of the SwISS algorithm applied to a
  two-dimensional Gaussian with three sub-posteriors (blue) approximating the full posterior (black). First, the batch samples are shifted by their respective means $(\btheta_b-\bmu_b)$, then scaled by the matrix $\bA_b$ before finally being shifted by the global mean $\bmu$.}
\label{fig:Grid_Exp}
\end{figure*}

Suppose that we have applied independent Monte Carlo algorithms, such as MCMC, in parallel to
sample from the inflated sub-posteriors \eqref{eq:inflate}, and denote
the $j^{\mathrm{th}}$ (of $J$) sample from the $b^{\mathrm{th}}$ (of $B$) batch by
$\theta_b^{(j)}$. The SwISS algorithm transforms each sample from
the inflated sub-posterior into a sample from an approximation to the full
posterior \eqref{eq:post} using a batch-specific affine
transformation. In the case of Gaussian sub-posteriors, as we show
below, these affine
transformations produce a set of samples from the correct Gaussian
full posterior, and SwISS is exact in this setting. In general, sub-posteriors are
non-Gaussian; however under standard regularity conditions and the Bernstein-von Mises theorem \citep{yang2000asymptotics}, as $n^*=n/B$  approaches infinity the sub-posteriors will be approximately Gaussian  and SwISS can be expected to produce samples
from an approximation to the full posterior.

Firstly, let us suppose that each inflated sub-posterior is Gaussian
with expectation $\bmu_b$ and invertible variance matrix $\bV_b$, so that the full posterior is $\btheta |  \by \sim \mathcal{N}_d(\bmu,\bV)$ where,
\begin{align}
  \label{eq:MCTpar}
  \bV = \left(\frac{1}{B}\sum_{b=1}^B \bV_b^{-1}\right)^{-1} \mbox{and}~ \ \bmu = \bV \frac{1}{B}\sum_{b=1}^B \bV_b^{-1}\bmu_b
\end{align}
are the variance and mean of the full posterior. 

Since it is invertible, $\bV_b$ is a positive-definite
matrix and therefore it has a $d\times d$, invertible square root,
$\bM_b$; \emph{i.e.} $\bV_b = \bM_b\bM_b^\top$. Similarly, the full posterior variance $\bV$ has a square
root, $\bM$, so that $\bV = \bM\bM^\top$. Let samples from the $b^{\mathrm{th}}$
inflated sub-posterior, $\btheta_b^{(1:J)}$, be (marginally) realisations from the random variable
$\btheta_b\sim \pi_b^B$ and define the transformed random variable:
\begin{equation}
  \label{eq:MCbasic}
  \bvtheta_b := \bA_b\left(\btheta_b-\bmu_b\right)+\bmu.
\end{equation}
where $\bA_b$ is any matrix satisfying $\bA_b \bV_b\bA_b^\top=\bV$; for
example, $\bA_b=\bM\bM_b^{-1}$.

Clearly,   $\Expect{\bvtheta_b} = \bmu$ and  $\Var{\bvtheta_b} = \bV$.
%\begin{align*}
%  \Expect{\bvtheta_b} = \bmu, &  & \text{and} & & \Var{\bvtheta_b} = \bV.
%\end{align*}
Furthermore, an affine transformation of a Gaussian random variable is
Gaussian, and hence $\bvtheta\sim \mathcal{N}_d(\bmu,\bV)$. Applying
the same transformation to individual
samples from the $b^{\mathrm{th}}$ batch, therefore provides a sample from the full posterior.
As discussed above, even when the sub-posteriors are not Gaussian, we can still apply the same scaling and shifting to any sub-posterior samples
and produce samples from an approximation to the full posterior.  

\subsection{Choice of Matrix Square Roots}
\label{sec:mat_sqr}

Matrix square roots are not unique; \emph{e.g.} for a diagonal
matrix, each element of the diagonal square root could be
negated; methods for finding a square root of a positive-definite matrix include the Cholesky decomposition, or the simple asymmetric square root arising from the spectral decomposition.
Moreover, for square roots $\bM$ and $\bM_b$,
$\bA_b$ need not be simply $\bM\bM_b^{-1}$, and indeed, this is not always the most sensible choice. 

For now, let $\bM$ be any $d\times d$ square root of $\bV$ and let
\begin{align*}
  \bVtil_b:=\bM^{-1} \bV_b \left(\bM^{-1}\right)^\top \quad \mbox{and} \quad \bA_b=\bM \bMtil^{-1}_b \bM^{-1},
\end{align*}
where $\bMtil_b$ is any $d \times d$ square root of $\bVtil_b$.
%\[
%\bVtil_b:=\bM^{-1} \bV_b \left(\bM^{-1}\right)^\top,
%\]
%let $\bMtil_b$ be any square root of $\bVtil_b$, and let
%\[
%\bA_b=\bM \bMtil^{-1}_b \bM^{-1}.
%\]
Then, $\Var{\bM^{-1}\btheta_b}=\bVtil_b$, so
$\Var{\bMtil^{-1}_b  \bM^{-1}\btheta_b}=\bI$ and hence $\Var{\bA_b \btheta_b}=\bM\bM^\top=\bV$.
%\begin{align*}
%  \Var{\bA_b\theta_b}&=
%  \bM \bMtil^{-1}_b \bM^{-1}~\bV_b~\left(\bM^{-1}\right)^T
%  \left(\bMtil^{-1}_b\right)^\top \bM^\top\\
%  &=
%  \bM \bMtil^{-1}_b \bVtil_b~
%  \left(\bMtil^{-1}_b\right)^\top \bM^\top\\
%    &=
%  \bM \bMtil^{-1}_b ~\bMtil\bMtil^\top~
%  \left(\bMtil^{-1}_b\right)^\top \bM^\top\\
%&=\bM \bM^{\top}=\bV.
%\end{align*}

If $\bV_b=\bV$ for all $b=1,\dots,B$, then $\bVtil_b=\bI$, and provided
$\bMtil_b$ is chosen to be $\bI$, $\bA_b$ becomes the
identity transformation. To be clear, though, if some diagonal
elements of $\bMtil_b$ had been chosen to be $-1$ rather than $1$ then
$\bA_b$ would not be the identity and, unless the initial distribution of points
$\btheta_b^{(1:J)}$ was elliptically symmetric, the transformation in
\eqref{eq:MCbasic} would not then lead to a set of points that
represented the true posterior at all.

Applying the same logic as above,
the transformation $\bMtil_b$ should be the square root of $\bVtil_b$
that moves the
individual points $\btheta_b^{(j)}$ as little as possible. With this in
mind, we define a natural measure of the distance moved by $J$ points,
$\btheta^{(1:J)}$, to which a linear transformation $\bA$ is applied, as:  
\begin{equation}
  \label{eq:MCdis}
  D\left(\bA; \btheta^{(1:J)}\right) := \frac{1}{J} \sum\limits_{j=1}^J \norm{\btheta^{(j)}-\bA\btheta^{(j)}}^2,
\end{equation}
where $\norm{.}^2$ denotes Euclidean distance. We
wish to find the linear transformation $\bMtil_b$ that minimises
$D(\bMtil_b^{-1}; \btheta^{(1:J)})$ subject to the constraint that
$\bMtil_b\bMtil_b^\top=\bVtil_b$. In Section \ref{subsec:MCbest} we
show that, provided the points have expectation zero, as
$J\uparrow \infty$, the best choice of $\bMtil_b$ is the positive-definite, symmetric square root of
$\bVtil_b$; this is the square root used by SwISS. The choice of
square root, $\bM$, of $\bV$ is less important, since within the linear
transformation $\bA_b$, the initial
transformation by $\bM^{-1}$ is later inverted; however, with the
general motivation of preventing excess movement, SwISS sets $\bM$ to be
the positive-definite, symmetric square root of $\bV$. Finally, the averaged re-centring algorithm of \cite{Wu2017} can be viewed as a special case of SwISS where $\bA_b = \bI$.

\subsection{The Positive-Definite, Symmetric Square Root and its Optimality}
\label{subsec:MCbest}
We first define the positive-definite symmetric square root of a
positive-definite matrix and detail the sense in which it is
optimal with respect to the distance measure $D$ \eqref{eq:MCdis}.

Let $\bV$ be a positive-definite matrix and let its  spectral decomposition be
\begin{equation}
  \label{eq:MCspec}
  \bV = \bU \bLambda \bLambda \bU^\top.
\end{equation}
where $\bLambda$ is a diagonal matrix with entries equal to the
positive square roots of the eigenvalues of $\bV$,
 and $\bU$ is a unitary matrix (\emph{i.e.} the columns of $\bU$ are the orthonormal right
eigenvectors of $\bV$, so $\bU \bU^\top=\bI=\bU^{\top}\bU$) and so
\begin{equation}
  \label{eq.symsqr}
\bV^{1/2} := \bM = \bU \bLambda \bU^\top.
\end{equation}
The natural interpretation of $\bM$ is as a simple scaling transformation
with different scalings applied along each of the eigenvectors of $\bV$.

As explained previously, we require a matrix $\bA_b$ such that the
transformation \eqref{eq:MCbasic} leads to a sample with a variance of
$\bV$; however, when inflated sub-posteriors are non-Gaussian, we need a transformation that preserves
the shape and orientation of the inflated sub-posterior as much as possible. 
Theorem \ref{thr:1} shows that for large $J$, $\bM^{-1}$ is not likely to cause more than the minimum discrepancy, given the constraints.

\begin{theorem}
\label{thr:1}
  Let $\btheta^{(j)} \in \mathbb{R}^d \left(j = 1,\dots,J\right)$ be a
  set of independent and identically distributed realisations of a random variable $\btheta$ with
  $\Expect{\btheta} =0$ and $\Var{\btheta}=\bV$. Let $\bV$ have a
  spectral decomposition as in \eqref{eq:MCspec} and let $\bM =
  \bU\bLambda \bU^\top$. Let $\bA$ be any other $d\times d $ matrix such that $\Var{\bA\btheta} = \bI_d$. Then
  \[
    \mathbb{P}\left( \underset{J\rightarrow \infty}{\lim} \left[D\left(\bM^{-1};\btheta^{(1:J)}\right)-D\left(\bA;\btheta^{(1:J)}\right)\right] > 0 \right) =0,
  \]
  where $D(\cdot;\cdot)$ is as defined in \eqref{eq:MCdis}.
\end{theorem}
% The proof of Theorem \ref{thr:1} is given in Section \ref{sec.proof} of the Appendix.

Theorem 1 relies upon the following two results.

\begin{proposition}
  Let $Z_i\in\mathbb{R}^d$ $(i=1,\dots,n)$ be an independent and identically distributed sequence of 
  random variables with $\Expect{Z}=0$ and $\Var{Z}=I_d$, and let $B$
  be any $d\times d$ matrix. Then, as
  $n\rightarrow \infty$, $D(B;Z_{1:n})\rightarrow
  d+\trace{BB^T}-2\trace{B}$, 
almost surely.
\end{proposition}

\begin{proof}
  \begin{align*}
  D(B;Z_{1:n})
  &=
  \frac{1}{n}\sum_{i=1}^n||(B-I)Z||^2
  \rightarrow
  \Expect{||(B-I)Z||^2} = \Expect{\sum_{i,j,k=1}^d Z_i (B-I)^T_{i,j}(B-I)_{j,k}Z_k}\\
  &=  \sum_{i,j=1}^d (B-I)^T_{i,j}(B-I)_{j,i} = \trace{(B-I)^T(B-I)}\\
  &= \trace{B^TB}+\trace{I}-2\trace{B},
  \end{align*}
  where the convergence is almost sure.
  But $\trace{B^TB}=\trace{BB^T}$, giving the required result.
\end{proof}

\begin{lemma}
Let $V,\Lambda,$ and $U$ be as defined in Theorem 1, and let and $\lambda_i=\Lambda_{i,i}$ $(i=1,\dots,d)$. Then
\[
\sup_{M:MM^T=V} \trace{M}=\sum_{i=1}^d\lambda_i.
\]
The supremum is achieved when $M=U\Lambda U^T$.
\end{lemma}

\begin{proof}
The rows of $U$ form an orthonormal basis $e_1,\dots,e_d$, with
$e_i^TVe_i=\lambda^2_i$ $(i=1,\dots,d)$. For any matrix $M$ with $MM^T=V$,
\[
\lambda_i^2=e_i^TMM^Te_i=||M^Te_i||^2=||f_i||^2,
\]
where $f_i=M^Te_i$. Next, recall that for any unitary matrix, $U$, and square matrix $M$, $\trace{UMU^T}=\trace{M}$. Thus, using the Cauchy-Schwarz inequality, and  since $U^T$ is also unitary,
\begin{align*}
\trace{M}
&=\trace{U^TMU}
=\sum_{i,j,k=1}^d(e_i)_j(e_i)_kM_{j,k}\\
&=
\sum_{i=1}^d e_i^T M e_i = \sum_{i=1}^d f_i^T e_i
                                           \le \sum_{i=1}^d||f_i||~||e_i|| =\sum_{i=1}^d\lambda_i.
\end{align*}
The final part of the Lemma follows as $\trace{U\Lambda U^T}=\trace{\Lambda}=\sum_{i=1}^d\lambda_i$.
\end{proof}

To prove Theorem 1,
let $Z_i=AX_i$, so $\Expect{Z_i}=0$ and $\Var{Z_i}=I_d$, and let
  $B=A^{-1}$ so $BB^T=V$. Since $D(A;X_{1:n})=D(B;Z_{1:n})$, by
  Proposition 1, and then from Lemma 1, we have almost surely,
  \begin{align*}
  D(A;X_{1:n})-D(M^{-1};X_{1:n}) &= D(B;Z_{1:n})-D(M;Z_{1:n})\\
  &\rightarrow 2\trace{M}-2\trace{B}\ge 0.  
  \qed
  \end{align*}

\begin{algorithm}[ht]
  \caption{SwISS Algorithm; here $\mathsf{SPSQ}(\bV)$ denotes the
    symmetric positive-definite square root of the matrix $\bV$ as
    described through \eqref{eq:MCspec} and \eqref{eq.symsqr}.}
  \label{alg:swiss}
  \begin{algorithmic}
    \Require \{$\btheta_b^{j}\}_{j=1}^{J}$ - $J$ Monte Carlo samples from each of the $B$ inflated posteriors \\
    \State Calculate the mean and variance for each of the inflated posteriors 
    \For{$ b \in \left\{1,\dots, B\right\}$}
    \State $ \bmu_b \leftarrow \mathsf{mean}[\btheta_b^{(1:J)}] \quad \mbox{and} \quad \bV_b \leftarrow \mathsf{var}[\btheta_b^{(1:J)}]$
    \EndFor \\
    \State Set the global mean $\bmu$ and variance $\bV$  and calculate the matrix square root, 
     \State $\bV = \left(\frac{1}{B}\sum_{b=1}^B \bV_b^{-1}\right)^{-1}, \quad  \bmu = \bV \frac{1}{B}\sum_{b=1}^B \bV_b^{-1}\bmu_b, \quad \mbox{and} \quad \bM \leftarrow \mathsf{SPSQ}(\bV)$ \\
%as in \eqref{eq:MCTpar} 
    \State Apply the affine transformation to the inflated posterior samples
    \For{ $b \in \left\{1, \dots, B \right\} $}
%    \State $\tilde{\btheta}_b^{(1:J)} \leftarrow
%    \bM^{-1}\left(\btheta^b-\hat{\bmu}_b\right) $
    \State $\bVtil_b\leftarrow\bM^{-1}\bV_b\bM^{-1}$
    \State $\bMtil_b \leftarrow \mathsf{SPSQ}(\bVtil_b)$
    \State $\bA_b\leftarrow \bM\bMtil_b^{-1}\bM^{-1}$
%    \State $\bz^b \leftarrow \tilde{\bU}\tilde{\bLambda}^{-1}\tilde{\bU}^\top\tilde{\btheta}^b$
    \State \mbox{Set} $\bvtheta_b^{1:J} \leftarrow \bA_b\left(\btheta_b-\bmu_b\right)+\bmu$ 
 %   \State $\leftarrow \bU\bLambda\bU^\top \bz^b + \hat{\bmu}$
    \EndFor \\
    \State Concatenate the transformed samples $\bvtheta_b^{1:J}$ to give a Monte Carlo approximation of the full posterior distribution $\pi(\btheta|\by)$
    \State \Return $\left\{\bvtheta_{1}^{(1:J)}, \dots, \bvtheta_{B}^{(1:J)} \right\}$  
  \end{algorithmic}
\end{algorithm}

The affine transformation \eqref{eq:MCbasic} of SwISS is easy to apply to each batch of inflated sub-posterior samples, making the algorithm as fast and as simple to use as the consensus algorithm, with the guarantee of exactness in the Gaussian case. A visual representation of SwISS is given in Figure~\ref{fig:Grid_Exp} and pseudo-code for implementing the algorithm is given in Algorithm~\ref{alg:swiss}.

\section{Experiments}
\label{sec:MCres}

In this section we test the accuracy of the SwISS algorithm to merge batch posterior samples drawn from a variety of posterior distributions. We consider various complex posterior geometries to highlight the difference between affine transformations of posterior samples (\emph{i.e.} SwISS) and averaging posterior samples (\emph{i.e.} Consensus Monte Carlo). We also investigate the efficiency of alternative merging algorithms on popular statistical models with simulated and real data. We compare the SwISS algorithm against the following popular competing algorithms from the literature:
\begin{itemize}
\itemsep0em
\item \textbf{Consensus Monte Carlo} (Cons) algorithm \citep{scott2016bayes}, as described in Section~\ref{sec:MCbac}.
\item \textbf{Semiparametric density estimation} (SKDE)\footnote{Implemented using the \textit{parallelMCMCcombine} R package}  from \cite{Neiswanger2014}, where sub-posteriors are approximated semi-parametrically as described in \cite{Hjort1995}.
\item \textbf{Average re-centring} (AR) algorithm from \cite{Wu2017}, which is a special case of SwISS where $\bA_b = \bI$.
\item \textbf{Gaussian Barycenter} (GB) algorithm \citep{Srivastava2015}, assuming a Gaussian approximation for each inflated sub-posterior, the barycenter is the geometric center of the inflated sub-posterior distributions.
% \item \textbf{LISA} the Likelihood inflating sampling algorithm \citep{LISA2016}, that uses the inflated sub-posteriors as importance sampling proposals.
 \end{itemize}

We assess the accuracy of the above algorithms to combine batch posterior samples to form an approximation of the full posterior, comparing the merged approximations against the full posterior, which is generated by sampling (in serial) from the posterior conditional on the full data set. Accuracy of estimation of the posterior of the $d$-dimensional parameter, $\btheta$, is assessed with the following discrepancy measures: 

\begin{itemize}
\itemsep0em
\item \textbf{Mahalanobis distance} (Mah):
  \[ D_{\mbox{Mah}} := \sqrt{\left(\bmu_a-\bmu_f\right)^\top\bV_f^{-1}\left(\bmu_a-\bmu_f \right)}, \]
where $\bV_f$ and $\bmu_f$ are the variance and mean estimates of posterior samples taken from the full data posterior using an MCMC algorithm. For a given posterior approximation algorithm, \emph{e.g.} SwISS, $\bmu_a$ denotes the estimated mean.

\item \textbf{Mean absolute skew deviation} (Skew):
$$\eta :=\frac{1}{d}\sum_{i=1}^d  |\hat{\gamma}^{a}_i - \hat{\gamma}^{f}_i|,$$ where 
  $\gamma_i =\mathbb{E}[\{(\btheta_i-\bmu_{i})/{\bV}^{1/2}_{ii}\}^3]$; \emph{i.e.} $\eta$ is the sum over components of the 
  third standardised moments.

\item \textbf{Integrated absolute distance} (IAD): 
  \begin{align*}
    D_{\mbox{IAD}} :=  \frac{1}{2d}\sum_{j=1}^d \int |\hat{\pi}^a_j(\theta_j) - \hat{\pi}_j^f(\theta_j)|\mathrm{d}\theta_j \in [0,1],
  \end{align*}
  the average of the integrated absolute differences between two kernel density estimates of the marginal posteriors for each component, $j$, of $\theta$: $\hat{\pi}_j^f$, obtained from samples from the true posterior, and  $\hat{\pi}^a_j$ using one of the  approximate merging algorithms \citep{chan2021divide}.

\end{itemize}

\subsection{Complex Posterior Geometries }
\label{subsec:MCgeom}

\begin{figure*}[t!]
  \centering
  \includegraphics[scale=0.3]{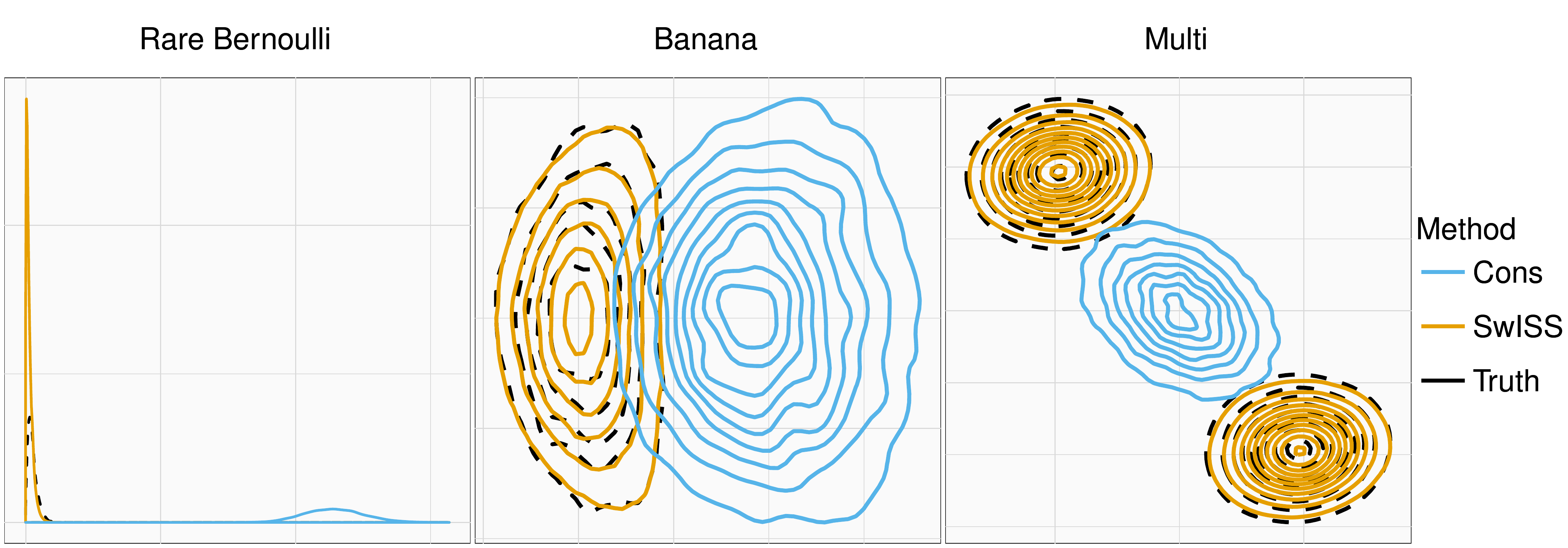}
  \caption{Density plots showing the posterior reconstructions using the SwISS algorithm against the Consensus algorithm on a rare Bernoulli target (left), warped Gaussian (also known as the banana-shaped target) (middle) and bi-modal target (right). In all cases the x-axis is $\theta_1$; for the left plot the y-axis is density, and for the other two plots it is $\theta_2$.}
  \label{fig:MCgeom}
\end{figure*}

One of the main motivations for using MCMC to sample from a posterior distribution, rather than using deterministic approximations (\emph{e.g.} Laplace), is that the posteriors are often non-Gaussian. We consider three artificially generated posterior distributions of dimension one or two (see Figure \ref{fig:MCgeom}) which reflect a range of potential posterior shapes and we compare  SwISS against the consensus Monte Carlo algorithm in these settings. Here $\phi(\mu)$ denotes the probability density function of a standard Gaussian $\mathcal{N}(\mu,1)$, for some $\mu \in \mathbb{R}$

\begin{itemize}
\item \textbf{Rare Bernoulli density}
  \[
    \pi_b(\theta_1 | \by_b) \propto \theta_1(1-\theta_1)^{999},   \quad \mbox{for each} \ b \in \{1,\dots, B\},
  \]
This corresponds to a posterior with 1000 Bernoulli observations with a single positive response and a uniform prior on the success probability, $\theta$, which gives a skewed posterior density. 

\item \textbf{Warped bivariate Gaussian density}
  \[
    \pi_b(\btheta | \by_b ) \propto \phi(\theta_1)\phi(\theta_2+\theta_1^2), \quad \mbox{for each} \ b \in \{1,\dots, B\},
  \]
 where $\btheta=(\theta_1,\theta_2)$. %Here each sub-posterior has a banana-like shape.
\item \textbf{Mixture of bivariate Gaussian densities}
  \[
    \pi_b(\btheta | \by_b ) \propto \phi(\btheta-\mu_1) + \phi(\btheta-\mu_2), \quad \mbox{for each} \ b \in \{1,\dots, B\},
  \]
  where $\btheta=(\mu_1,\mu_2)$. 
\end{itemize}

Both the SwISS and the consensus algorithm are guaranteed to be exact in the case of merging Gaussian posterior samples, but it can be shown that both algorithms still work well for a variety of non-Gaussian posteriors. However, one of the drawbacks of the consensus algorithm is that averaging across batches of sub-posterior samples can remove posterior features such as skewness and multi-modality, as illustrated in Figure~\ref{fig:MCgeom}. 

Figure~\ref{fig:MCgeom} shows posterior density plots for each of the three models, where full MCMC has been utilised to provide a \textit{ground truth} approximation for the full data posterior. The consensus Monte Carlo and SwISS approximations are based on combing samples from $B=10$ sub-posterior and inflated sub-posterior approximations, respectively. The results from these three test cases show that the consensus algorithm struggles to approximate the full data posterior when the target density exhibits non-Gaussian behaviours. The SwISS algorithm, which utilises affine transformations of the inflated sub-posterior samples, rather than averaging, can produce reliable approximations when the posterior is significantly non-Gaussian.

\subsection{Scalability with parameter dimension}
\label{sec:scalability}

Typically, divide-and-conquer methods are advertised for use with tall data, \emph{i.e.} a large number of observations and up to a moderate number of parameters. Here, we test the accuracy and computational speed of the merging algorithms as the number of parameters grows.

Let $\btheta | \by_b  \sim \mathcal{N}_d(\bmu_b,\bs{V}_b)$ for $b \in \left\{1,\dots, B=10\right\}$, where $d$ is the dimension of the parameter space and let $\bmu_b \sim \mathcal{N}_d(0,\mathrm{I}_d)$ and $\bs{V}_b \sim \mathcal{W}^{-1}(5d, \mathrm{I}_d)$. Each sub-posterior is Gaussian, with expectation and variance drawn respectively from Gaussian and inverse-Wishart distributions. For each experiment $J=5,000$ samples were drawn from each sub-posterior and inflated sub-posterior. Using this model, the full data posterior is tractable: 
\[
 \btheta | \by \sim \mathcal{N}_d\left( V \sum\limits_{b=1}^B V_b^{-1}\bmu_b,V \right),
 \]
 where $V^{-1} = \sum_{b=1}^B V_b^{-1}$. The following set of dimensions were used: $d \in \left\{ 5, 10, 20, 40, 80 \right\}$.

 Figure~\ref{fig:timings_scale} shows that both the consensus Monte Carlo algorithm and SwISS perform well with increasing dimension (as measured by integrated absolute distance) and are both computationally efficient. The semi-parametric KDE approach, Gaussian barycenter and average re-centering approaches display reduced accuracy (as measured by integrated absolute distance). Only SwISS and consensus are robust to increasing the dimension of the parameter space. In terms of the computational cost required to merge the posterior samples, all approaches are generally fast, with the exception of the semi-parametric KDE approach. 

 \begin{figure}[t!]
\centering
\includegraphics[scale=0.4]{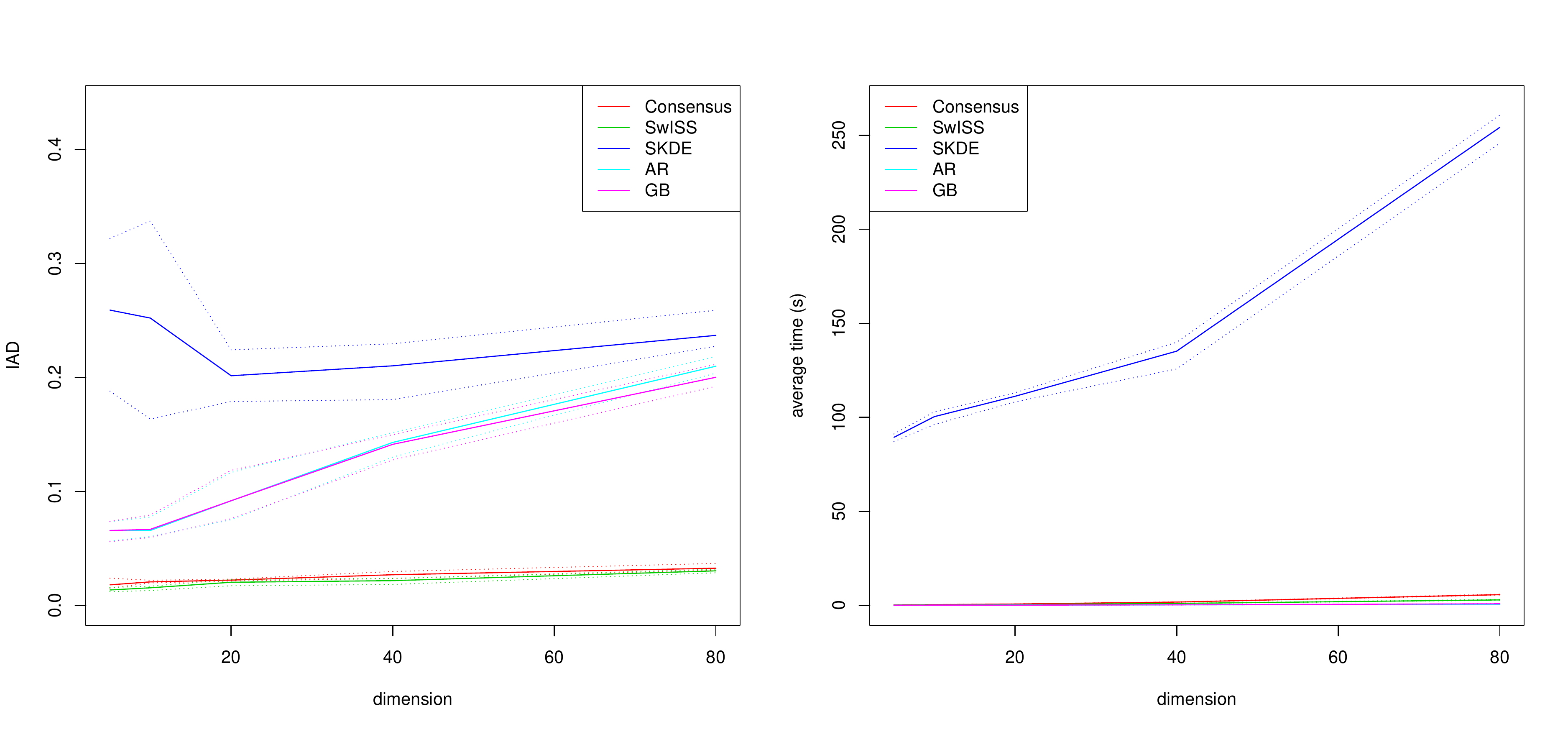}
\caption{The left plot shows integrated absolute distance for each method for a different number of parameters. The right plot shows the mean time each combination method took to obtain the samples from an approximate posterior.}
\label{fig:timings_scale}
\end{figure}

\subsection{Linear Mixed Effects Model}
\label{sec:linear-mixed-effects}

A natural way to extend the simple linear model is to introduce both fixed and random effects. This extension can be particularly useful when data exhibit a hierarchical dependency structure, for example, to cluster student test scores based on classroom. Let  $y_{i,j} \in \mathcal{Y} \subseteq \mathbb{R}$ (for $i, j=1,\ldots,n_j$, and $j=1,\ldots,n$) be the response variable, where $n_j$ is the number of observations for group $j$. The fixed and random effects are $\bx_{i,j} \in \mathcal{X} \subseteq \mathbb{R}^{p}$ and $\bz_j \in \mathcal{Z} \subseteq \mathbb{R}^{r}$, respectively, and are related to the response variable by
$$ y_{i,j} | \bbeta, \balpha_j, 1 \sim \text{Logistic }(\bx_{i,j}\bbeta + \bz_j\balpha_j, 1), \ \ \balpha_i \sim \mathcal{N}_r(0,\Sigma),$$
where $\bbeta \in \mathbb{R}^p$ and $\balpha_i \in \mathbb{R}^r$ are the fixed and random effect model coefficients and the $\mathsf{Logistic}(a,1)$ distribution has a cumulative distribution function of $1/(1+\exp[-(x-a)])$. Our parameters of interest are then $\btheta=(\bbeta,\balpha,\Sigma)$, where $\Sigma$ represents the variance of the random effects. We assume an inverse-Wishart distribution for the prior of $\Sigma$,
 $\Sigma \sim \mathcal{W}^{-1}(\nu,S) $, with $\nu=5$ and $S=5I_r$, and \emph{a priori} we assumed $\bbeta\sim \mathcal{N}_p(0,1000 I_p)$.

%We test the accuracy of the posterior merging algorithms using the MovieLens\footnote{https://grouplens.org/datasets/movielens/} data set on the linear mixed effects model. 
%We simulated a dataset that contains $100 \times 2000$ movie ratings, where users are asked to rate movies with a score between $0.5$ and $5$, with increments of $0.5$. As well as the rating information, each movie is assigned a genre category, with 19 genres in total. Using the ratings as the response variable, we fit the linear mixed effects model to these data using movies and users as feature variables. We used $p=6$ fixed-effects and $r=6$ random-effects and follow the procedure of \cite{Srivastava2015} to reduce the 19 genres to 4 categories. The data are randomly partitioned into $B=10$ batches. As in the previous example, we used the STAN software to sample from the full posterior and sub-posteriors generating $J=5,000$ samples.

We simulated a dataset that contains $200,000$ observations. We set the number of groups $n=2000$ and the number observations for each group $n_j = 100$, for $j=1,\ldots,n$. The number of parameters for the fixed effects were set to $p=10$, with $\beta_{0,i} = (-1)^{(i -1)}$. The number of parameters for each random effect was set to be $r=2$, and we set  
$$\Sigma_0 = \left[\begin{array}{cc}1 & 0.2 \\ 0.2 &1 \end{array} \right], $$ then $\alpha_i$ were simulated independently from a $\mathcal{N}_r(0,\Sigma_0)$ distribution. We included an intercept term, that is $x_{i,j,1}=1$ for all $i,j$, otherwise $\bx_{i,j}$ and $z_{j}$ were simulated from independent $\mathsf{Bernoulli}(0.5)$ distributions. 

The data were randomly partitioned into $B=10$ batches by group, so that each group only belonged to one batch. This was necessary since divide and conquer methods assume independence between the batches. With a Gaussian observation model for the $y_{i,j}$, marginalisation over all of the random effects would be tractable. The logistic observation model necessitates the use of a sampling scheme such as MCMC.

We used the STAN software to sample from the full posterior and sub-posteriors generating $J=5,000$ MCMC samples after an initial 1,000 sample burn in. Table \ref{tab:MCresLME} gives the discrepancy measures for each of the merging algorithms, averaged over 10 random partitions of the data. The results show that all algorithms perform well, with the exception of AR and the Gaussian barycenter, both on the Mahalanobis metric. The SwISS and Consensus algorithms are robust across the range of metrics. 
\begin{table}[t]
  \centering
  \caption{Discrepancy measures for the linear mixed effects model on the simulated dataset. These measures were averaged over 5 random partitions of the data. Estimated standard errors are given in brackets.
  }
  \label{tab:MCresLME}
  \begin{tabular}{l||p{1.8cm}p{1.8cm}p{1.8cm}}
    Algorithm  & Mah & Skew & IAD \\
    \hline 
    SwISS  & 0.69 (0.14)  & \textbf{0.02} ($<$0.01)  & 0.06 (0.01)   \\
    Consensus   & \textbf{0.39} (0.13)  & 0.03 (0.01)   & \textbf{0.04} (0.01)  \\
    Average Re-centring     & 2.20 (0.31)  & \textbf{0.02} ($<$0.01)  & 0.13 (0.01) \\
    Semi-parametric KDE   & \textbf{0.39} (0.08)  & 0.04 ($<$0.01)  & \textbf{0.04} (0.01)  \\
    Gaussian Barycenter & 2.19 (0.31) & 0.05 (0.01) & 0.13 ( 0.01)
  \end{tabular}
\end{table}

\subsection{Logistic Regression Model}
\label{subsec:MClogis}

Logistic regression is a popular technique for modelling binary data, \emph{i.e.} $y_i \in \{0,1\}$. Features $\bx_i \in \mathbb{R}^{d}$, also known as covariates, that can indicate the classification outcome are mapped onto the binary observations using a logit transformation, where the outcome probability $\mathbb{P}(y_i=1) = \exp(\bx_i^\top\btheta)/\left(1+\exp(\bx_i^\top\btheta)\right)$, is the success probability of a Bernoulli random variable. Our parameter of interest $\btheta \in \mathbb{R}^{d}$ is the vector of coefficients.

We consider two data sets, the first is a synthetic data set which is designed to simulate a scenario with rare but highly informative features. This data set is similar to the one given in \cite{scott2016bayes}. We simulate $N=100,000$ data points with $d=5$ binary features with relative frequencies of $\mathbf{x}_i=1$ being $( 1, 0.02,0.03,0.05,0.001)$ and the corresponding true parameter values are $\btheta_0=( -3,1.2,-0.5,0.8,3)$.  Due to the rarely occurring final feature, this can lead to largely differing variances across the sub-posteriors. For our experiments, we spilt the data equally across $ B=25$ batches.

We also consider a real-world data set; the Hepmass data set\footnote{http://archive.ics.uci.edu/ml/datasets/HEPMASS} from high-energy particle physics where the response is an indicator for whether a signal was indicative of an exotic particle being present as opposed to background noise. The data set contains 27 real features which we augmented with an intercept term to give $d=28$ parameters. The full data set contains 10.5 million responses, in this experiment we considered the first $N=100,000$ and split the data across $B=20$ batches. 

In each of the our experiments, the data were repeatedly partitioned $n_{\text{runs}}=5$ times with a Monte Carlo average of the discrepancy metrics given in Table \ref{tab:MCres}. The  STAN \citep{STAN} software, which implements an automatically-tuned version of Hamiltonian Monte Carlo sampling, was used as the MCMC sampler and applied to the full posterior and sub-posteriors for each experiment. Each sampler drew $J=10,000$ samples after a burn in of $1,000$ iterations. 

The results in Table \ref{tab:MCres} show that SwISS and the Consensus algorithm outperform all of the others on the simulated data, whereas for the real example all of the methods work well. The AR algorithm performs especially poorly on the synthetic data example as the variance of the sub-posteriors varies across subposteriors, and the AR algorithm does not correct for this when the sub-posterior samples are merged. The similarity of eprformance on the Hepmass data could be due to the sub-posteriors all being close to Gaussian. We would expect both Consensus and SwISS to work well in this setting as they are exact for Gaussian sub-posteriors, and much faster to apply than nonparametric methods such as SKDE.

\begin{table}[t]
  \centering
  \caption{Discrepancy measures for the logistic regression model with simulated and Hepmass data sets. Each metric was averaged over 5 runs. Estimated standard error are given in brackets.
  }
  \label{tab:MCres}
  \begin{tabular}{l||p{1.7cm}p{1.7cm}p{1.8cm} | p{1.8cm}p{1.8cm}p{1.8cm}}
      & \multicolumn{3}{c}{Simulated data}  & \multicolumn{3}{c}{Hepmass data}  \\
    Algorithm  & Mah & Skew & IAD  & Mah & Skew & IAD  \\
    \hline
     SwISS  & \textbf{0.46} (0.30) & \textbf{0.04} (0.01)   & \textbf{0.05} (0.03) & 0.56 (0.05)& \textbf{0.03} ($<$0.01) & 0.03 ($<$0.01)   \\
         Consensus   & 0.48 (0.35) & 0.05 (0.01)   & 0.06 (0.03)  & \textbf{0.35} (0.05)    & \textbf{0.03} ($<$0.01)     &  0.03 ($<$0.01)   \\
         Average Re-centring     & 5.46 (3.92) & 0.13 (0.07)  & 0.20 (0.03)  & 0.47 (0.04)   & \textbf{0.03}($<$0.01)     & \textbf{0.02} ($<$0.01)   \\
         Semi-parametric KDE   & 1.25 (1.11) & 0.76 (0.33)   & 0.12 (0.06)  & 0.36 (0.05)   & \textbf{0.03}($<$0.01)      &  0.03 ($<$0.01) \\
        Gaussian Barycenter & 5.42 (3.75) & \textbf{0.04} (0.01) & 0.20 (0.01) & 0.46 (0.05) & \textbf{0.03}($<$0.01) & \textbf{0.02} ($<$0.01)

  \end{tabular}
\end{table}

\section{Conclusions}
\label{sec:conclusions}

We have introduced a new method to merge posterior samples generated in parallel on independent batches of data. Our algorithm, SwISS, is fast, scalable to high-dimensional settings, and accurate on a variety of test cases. The SwISS algorithm, like the consensus Monte Carlo algorithm, is simple to apply and competitive against popular alternative divide-and-conquer algorithms. SwISS also has the advantage that it does not require hyper-parameter tuning and is faster to apply than many of the alternative divide-and-conquer algorithms given in the literature. We have provided theoretical support for our choice of affine transformations and shown that SwISS is exact in the case of merging inflated Gaussian sub-posteriors. Code to recreate this work is available through the Github link: \url{https://github.com/CJohnVyner/SwISS}

\subsection{Acknowledgements} 
The authors gratefully acknowledge the support of the UK Engineering and Physical Sciences Research Council grants EP/S00159X/1, EP/V022636/1 and EP/P033075/1.

\bibliographystyle{apalike}
\bibliography{refs}

%\nocite{*}% Show all bib entries - both cited and uncited; comment this line to view only cited bib entries;

\end{document}